\documentclass[aps,twocolumn,superscriptaddress,dvipsnames,nofootinbib,amssymb,floatfix,10pt,prx]{revtex4-2} 
\usepackage{graphicx}
\usepackage{dcolumn}
\usepackage{bm}
\usepackage{qcircuit}
\usepackage{braket}
\usepackage{color}
\usepackage{amsthm}
\usepackage{amsmath}
\usepackage{chemarrow}
\usepackage{overpic}
\usepackage{diagbox}
\usepackage{braket}
\usepackage{subfigure}
\usepackage{makecell}
\usepackage{tikz}
\newtheorem{theorem}{Theorem}

\newtheorem{lemma}{Lemma}

\newtheorem{prop}{Proposition}
\usepackage{chngcntr}
\usepackage[colorlinks=true,linkcolor=red,bookmarks=true,breaklinks=true,citecolor=blue]{hyperref}
\usepackage{cleveref}

\renewcommand{\addcontentsline}[3]{}
\newcommand{\Tr}{\mathrm{Tr}}
\newcommand{\re}{\text{Re}}
\newcommand{\im}{\text{Im}}

\begin{document}
\title{Hamiltonian dynamics from pure dissipation}
\author{Zhong-Xia Shang}
\email{zhongxia.shang@math.ku.dk}
\affiliation{Department of Mathematical Sciences, University of Copenhagen, Universitetsparken 5, 2100 Denmark}
\author{Daniel Stilck França}
\email{dsfranca@math.ku.dk}
\affiliation{Department of Mathematical Sciences, University of Copenhagen, Universitetsparken 5, 2100 Denmark}

\begin{abstract}
\noindent 
The fundamental difference between closed and open quantum dynamics lies in their environmental interaction: closed systems are perfectly isolated and evolve reversibly under unitary Hamiltonian dynamics, whereas open systems continuously couple to an external bath, resulting in irreversible dissipation and information loss. In this work, we show internal Hamiltonian dynamics can be "faked`` via external pure dissipation, i.e., Lindbladians without a coherent Hamiltonian part. More concretely, we show that, in a GKSL representation with zero explicit Hamiltonian term but nontraceless jump operators, bounded-norm dissipative generators can approximate Hamiltonian dynamics within $\epsilon$ error in diamond norm using $\mathcal{O}(t^2/\epsilon)$ evolution time. We further prove that for time-independent dynamics this $\mathcal{O}(t^2/\epsilon)$ scaling is in the worst case, necessary and optimal from a geometric perspective, which captures the fundamental decoherence cost for catching up with the speed of Hamiltonian dynamics. Our construction leads to various implications, including the BQP-completeness of purely dissipative dynamics even before reaching approximate equilibrium, a Zeno-adjacent state-independent freezing effect, the no super-quadratic fast-forwarding theorem of a class of purely dissipative dynamics, and reducing Lindbladian simulation cost via gauge changing. 

\end{abstract}
\maketitle
\noindent\textit{\textbf{Introduction.—}}The dynamics of any physical systems are governed by a combination of internal interactions and external coupling to an environment. For a quantum system described by a $d$-dimensional Hilbert space $\mathcal{H} \cong \mathbb{C}^d$, the most general trace-preserving and completely positive Markovian evolution is described by the Gorini--Kossakowski--Sudarshan--Lindblad master equation (Lindbladian)~\cite{gorini1976completely,lindblad1976generators}
\begin{equation}\label{mainlme}
\frac{d\rho}{dt}=-i[H,\rho]+ \sum_{i} \left(F_i\rho F_i^\dag- \frac{1}{2}\{\rho,F_i^\dag F_i\}\right).
\end{equation}
Here, the first term captures the internal, coherent evolution driven by the system Hamiltonian $H \in \mathbb{C}^{d \times d}$, while the second term captures the external interactions mediated by the jump operators $\{F_i\} \subset \mathbb{C}^{d \times d}$. When external interactions are completely absent, the system is closed; it evolves unitarily and reversibly, preserving all quantum information. Conversely, the inclusion of the jump operators defines an open quantum system, where continuous interaction with a bath drives irreversible dissipation and decoherence~\cite{breuer2002theory}.

Understanding how to simulate quantum systems and the effect of dissipation or decoherence is a central theoretical pursuit in quantum mechanics, with profound practical implications for quantum information and quantum computing. Especially, a recent wave of dissipation-based quantum algorithms has emerged, spanning Gibbs state~\cite{chen2025efficient,rouze2026optimal,chen2024local} and ground state~\cite{cubitt2023dissipative,ding2024single,zhan2026rapid} preparation, solving differential equations~\cite{shang2025designing}, and optimization~\cite{chen2025quantum}. The performance, advantages, and inherent limitations of these dissipative approaches, especially when compared to traditional coherent algorithms~\cite{dalzell2023quantum}, are deeply rooted in how open and closed dynamics interface and diverge~\cite{chen2024local}.

To interrogate this boundary, we ask: \textit{can coherent Hamiltonian dynamics be approximately reproduced by purely dissipative interactions, and if so, at what cost?} While Stinespring dilation~\cite{nielsen2010quantum} allows open systems to be modeled via closed dynamics, it does so by artificially expanding the Hilbert space to encompass the entire bath, thereby redefining the system-environment partition and converting external interactions into internal ones. However, if we restrict our view to the system alone, the fundamentally distinct physical origins and time-reversal properties of these two mechanisms suggest they cannot substitute for one another. Even though the gauge invariance~\cite{breuer2002theory} of the Lindbladian allows a certain level of freedom of interconversions, finding $F_i$ such that
\begin{equation}\label{eqvi}
\sum_{i} \left(F_i\rho F_i^\dag- \frac{1}{2}\{\rho,F_i^\dag F_i\}\right)= -i[H,\rho]
\end{equation}
holds exactly is impossible for any nonzero $H$, even if we allow for $F_i$ that have a nonzero trace~\cite{hayden2022canonical}.

In this work, we demonstrate that the internal, coherent Hamiltonian interactions can be well mimicked by strictly external, dissipative ones. Specifically, we prove that purely dissipative Lindbladians (i.e. Lindbladians with only $\{F_i\}$ and no $H$ in \cref{mainlme}) can achieve approximate ideal $t$-time Hamiltonian evolution within an error $\epsilon$ in diamond norm using an evolution time of $\mathcal{O}(t^{2}/\epsilon)$. This scaling shows that this imitation is not only possible but also efficient with no exponential overheads in time, system size or precision. Since the purely dissipative dynamics is fundamentally irreversible, we further show that this $\mathcal{O}(t^{2}/\epsilon)$ Lindbladian evolution time is actually necessary and optimal to simulate the reversible Hamiltonian dynamics in the worst case. We prove this via geometry: how to maximize the rotational speed (Hamiltonian dynamics) while minimizing the radial one (decoherence). Later, we show several implications of our results for various topics, including the complexity class of dissipative dynamics, Zeno effect, fast-forwarding dissipative dynamics, and Lindbladian simulations.
\begin{figure}[htbp]
\centering
\includegraphics[width=0.49\textwidth]{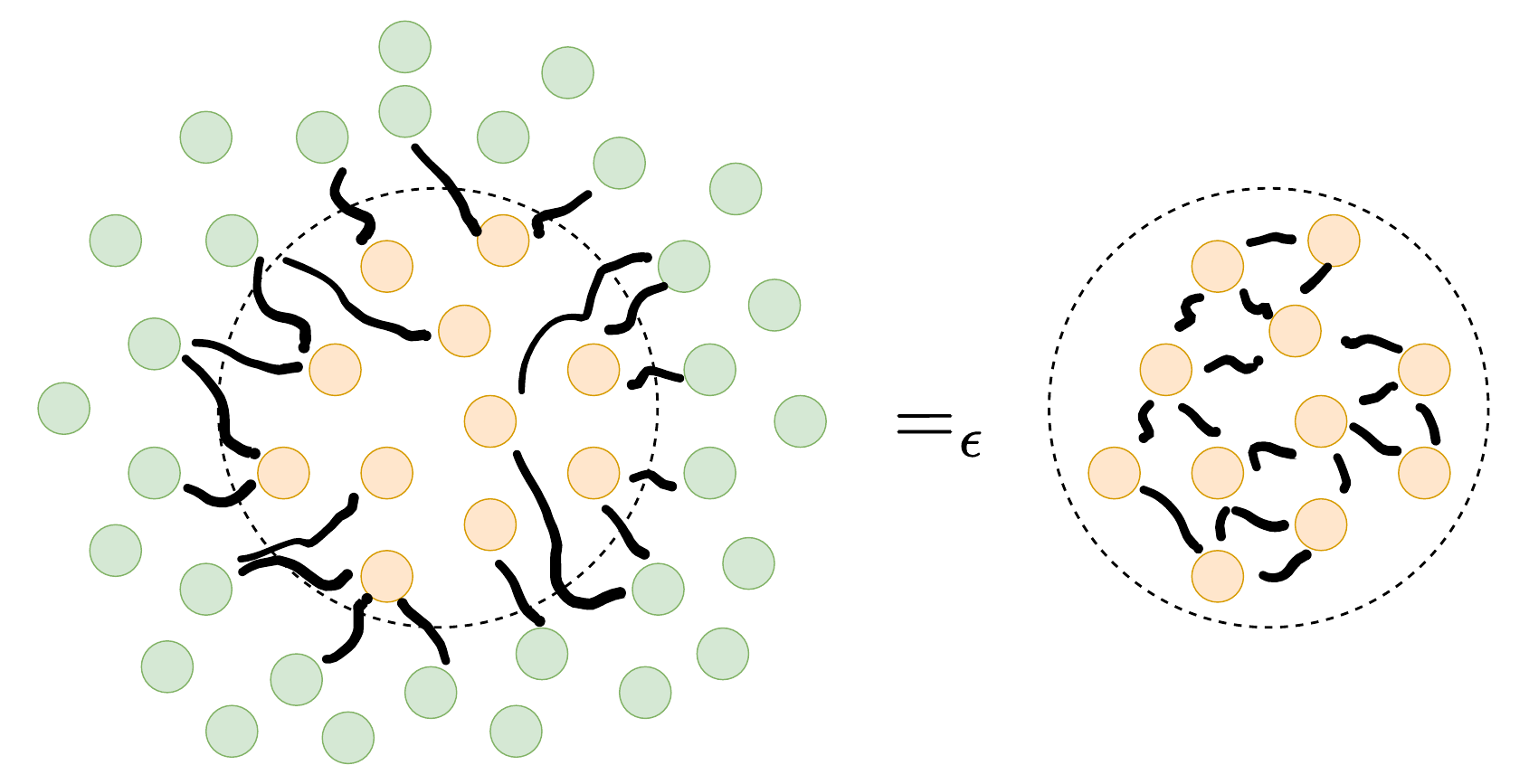}
\caption{The physical meaning of \cref{the1}. Our construction essentially says that for any closed quantum system characterized by the Hamiltonian $H$ with zero external dissipation, we can always find a purely dissipative open quantum system with zero internal Hamiltonian to efficiently mimic its dynamics.\label{fig1}}
\end{figure}

\noindent\textit{\textbf{Hamiltonian dynamics via dissipation.}} Our construction is simple and direct. Consider a purely dissipative, single-jump Lindbladian
\begin{align}\label{slme}
\frac{d\rho}{dt}&=\mathcal{L}[\rho]=F \rho F^\dag-\frac{1}{2}\{F^\dag F,\rho\},
\end{align}
i.e. in GKSL form with only one jump (albeit not necessarily traceless). If we set $F=I-i\delta H$, then we will have
\begin{align}\label{dmis}
\frac{d\rho}{dt}=-i\delta [H,\rho]+\delta^2 \left(H\rho H-\frac{1}{2}\{H^2,\rho\}\right),
\end{align}
where the first term is a $\delta$-strength Hamiltonian dynamics and the second term is $\delta^2$-strength dissipation when viewing $H$ as a jump operator, which has recently been found to have deep connections with quantum phase estimation~\cite{shang2025fast}. Based on \cref{dmis}, we obtain the following theorem.
\begin{theorem}\label{the1}
By setting $\delta=\epsilon/t$, the purely dissipative, single-jump Lindbladian \cref{slme} can approximate the Hamiltonian dynamics $e^{\mathcal{L}_Ht}[\rho]=e^{-iH t}\rho e^{iH t}$ within $\epsilon$ error in diamond distance~\cite{kitaev2002classical} in the sense that $\sup_{0\le s\le t}\left\|e^{\mathcal{L}\frac{T}{t}s}-e^{\mathcal{L}_Hs}\right\|_\diamond\le \epsilon$, by an evolution time
\begin{align}
T=\mathcal{O}\left(\|H\|_\infty^2t^2\epsilon^{-1}\right).
\end{align}
\end{theorem}
\noindent We give detailed proofs in App.~\ref{proof}, which uses Duhamel's Identity~\cite{evans2022partial} to bound the distance between dynamics via the distance between the Liouvillian operators.  The \cref{the1} is surprising in terms of the physical picture as shown in \cref{fig1}. Closed quantum systems have zero dissipation and non-zero internal Hamiltonians, while purely dissipative open quantum systems have zero internal Hamiltonians but non-zero dissipation. Therefore, these two types of systems are exactly the opposite of each other and should have fundamentally different dynamics. However, our result here shows that in fact the purely dissipative one can efficiently mimic the other with only an overhead that matches the first-order product formula-based Hamiltonian simulation method~\cite{lloyd1996universal}.

The design of the jump operator is not unique, and we show two generalizations. First, as long as $F=I-\delta A+\mathcal{O}(\delta^2)$ and the anti-Hermitian part of $A$ is $H$, then the Lindbladian evolution time to achieve the same goal in \cref{the1} requires
\begin{align}\label{gen1}
T=\mathcal{O}\left(\|A\|_\infty^2t^2\epsilon^{-1}\right).
\end{align}
A direct example of such a jump is $F=e^{-iH\delta}$, which is unitary short-time Hamiltonian dynamics itself. 

Second, one may argue $F=I-i\delta H$ is not physical since $H$ is a global operator, whereas physical jumps should be local. But note that we can actually localize jumps by increasing their number. Suppose $H$ has the decomposition $H=\sum_i H_i$, with $H_i$ local, then for each term, we can assign a jump operator $F_i=I-i\delta H_i$, and the resulting multi-jump purely dissipative Lindbladian can approximate the Hamiltonian dynamics by an evolution time
\begin{align}
T=\mathcal{O}\left(\sum_i\|H_i\|_\infty^2t^2\epsilon^{-1}\right).
\end{align}
This version is especially relevant for platforms with local dissipative control~\cite{harrington2022engineered}, and it shows that coherent many-body dynamics can be generated from local open-system primitives with only polynomial overhead.

One may notice that in all the above constructions, the jumps have the form $I+\delta O+\mathcal{O}(\delta^2)$ with $O$ some operator. In fact, we can prove this is the only possible form for smooth, purely dissipative Lindbladians that mimic the Hamiltonian dynamics. 

\begin{prop}[First-order uniqueness]\label{mainprop:rigidity}
\label{prop:first-order-uniqueness}
Let $\mathcal L_{D,\delta}
    = \sum_{j=1}^{m}F_{j,\delta}\rho F_{j,\delta}^\dagger-\frac12\{F_{j,\delta}^\dagger F_{j,\delta},\rho\}$ be a purely dissipative Lindbladian on a finite-dimensional
system, with fixed $m$ and jumps depending smoothly on $\delta$
at $\delta=0$. Let
$\mathcal L_H=-i[H,\cdot]$. Suppose that, for some $t_0>0$,
\begin{equation}\label{consss}
    \sup_{0\le s\le t_0}
    \left\|
        e^{(s/\delta)\mathcal L_{D,\delta}}
        -e^{s\mathcal L_H}
    \right\|_\diamond
    \xrightarrow{\delta\to0^+}0.
\end{equation}
Then every jump necessarily has the form
\begin{equation}
F_{j,\delta}=c_jI+\delta G_j+O(\delta^2).
\end{equation}
\end{prop}
We refer to App.~\ref{app:uniqueness} for the detailed proof. But note that our proof boils down to showing that Eq.~\eqref{consss} implies constraints on how the family of smooth Lindbladians $\mathcal{L}_{D,\delta}$ must behave at different orders of $\delta$.

An interesting remark here is that while shifting the Hamiltonian (adding identity) is trivial and only introduces a global phase, shifting the jump operator can fundamentally change the dynamics.

We finish this section by introducing a corollary of \cref{the1}. Since Hamiltonian dynamics can be recast as dissipation, therefore, we can actually convert any Lindbladian
into a pure dissipation form. For the general form \cref{mainlme}, we can construct a purely dissipative Lindbladian with jump operators $F'$ where we let $F'_0=I-i\delta H$ and let $F_i'=\sqrt{\delta}F_i$, then by setting $\delta=\epsilon/t$, to simulate the original Lindbladian for an evolution time $t$ within error $\epsilon$, the required new Lindbladian evolution time is
\begin{align}
T=\mathcal{O}\left(\|H\|_\infty^2t^2\epsilon^{-1}\right).
\end{align}

\begin{figure}[htbp]
\centering
\includegraphics[width=0.49\textwidth]{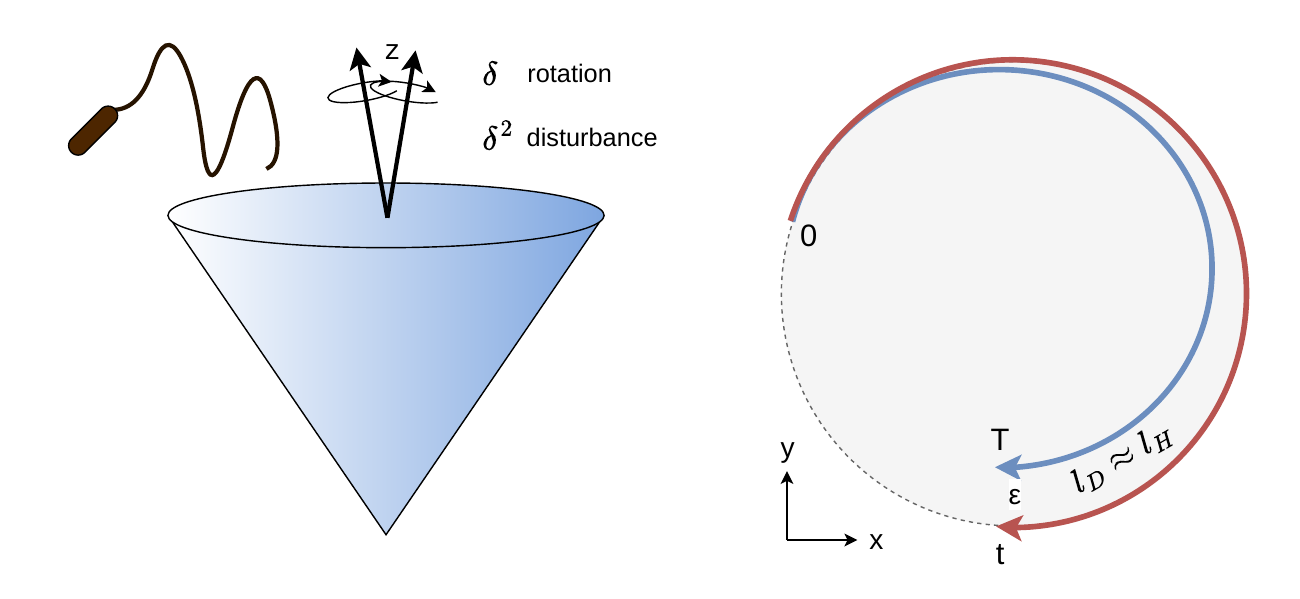}
\caption{\textbf{Left: }Hamiltonian dynamics via dissipation can be compared to driving a top via a whip. Since we want to rotate the top (Hamiltonian dynamics) via an external whip (environment), the rotation of a $\delta$-strength is inevitably accompanied by radial disturbance (dissipation) of $\delta^2$-strength. \textbf{Right: }Under the uniform-in-time error restriction \cref{uerr}, the dissipative trajectory (blue) should be close to the Hamiltonian trajectory (red). As results, the radial distance between the ideal state and the Lindbladian evolved state should be within $\epsilon$, and the trajectory length of the dissipative dynamics ($l_D$) should be close to the trajectory length of the Hamiltonian dynamics ($l_H$).\label{fig2}}
\end{figure}

\noindent\textit{\textbf{Geometric Lower Bound.}} Now, we want to show that since the dissipation is irreversible and fundamentally different from the Hamiltonian dynamics after all, the $\mathcal{O}(t^2/\epsilon)$ cost is what you have to pay in the worst case, as illustrated in the theorem below:
\begin{theorem}\label{the2}
There exist families of Hamiltonians $H\in\mathbb{C}^{2\times 2}$ such that 
if a purely dissipative Lindbladian of the form
\begin{equation}\label{dlm}
\mathcal{L}_D(\rho)=\sum_{j=1}^m\left(F_j\rho F_j^\dagger-\tfrac12\{F_j^\dagger F_j,\rho\}\right),
~~ \|F_j\|_\infty\le C,
\end{equation}
can simulate the target Hamiltonian dynamics for an evolution time $t$ within an diamond-distance error $\epsilon$ uniformly in time
\begin{equation}\label{uerr}
\sup_{0\le s\le t}\left\|e^{\mathcal{L}_D\frac{T}{t}s}-e^{\mathcal{L}_Hs}\right\|_\diamond\le \epsilon,
\end{equation}
then the total Lindbladian evolution time $T$ must satisfy the lower bound
\begin{equation}\label{lT}
T= \Omega\left(\frac{\|H\|_\infty^2}{m C^2}t^2\epsilon^{-1}\right).
\end{equation}
\end{theorem}

We have two remarks for \cref{the2}. First, the $mC^2$ factor in \cref{lT} matches the block encoding cost in the digital Lindbladian simulation algorithms on quantum computers~\cite{kliesch2011dissipative,childs2016efficient,cleve2016efficient,li2022simulating,yu2025lindbladian,peng2408quantum,ding2024simulating,pocrnic2025quantum,kato2024exponentially}. Second, the reason we ask the uniform-in-time error assumption is to avoid the adversarial cheating where the Lindbladian doesn't actually simulate the Hamiltonian dynamics but only approximates the unitary map $e^{\mathcal{L}_H t}$ at the specific time $t$.

The proof idea of \cref{the2} is from a geometric point of view. Since the lower bound cares about the worst case, we can analyze the lower bound via a single-qubit system where we consider approximating $H=\nu\sigma_z/2$ via \cref{dlm}. By mapping the qubit's density matrix to a Bloch vector $r \in \mathbb{R}^3$, the purely dissipative Lindbladian \cref{dlm} can be recast as an affine Bloch equation of the form $\dot{r} = Br + c$~\cite{nielsen2010quantum} with 
\begin{equation}
B = S + [\boldsymbol{\omega} \times].
\end{equation}
The matrix $[\boldsymbol{\omega} \times]$ (representing cross product by $\boldsymbol{\omega}$) is real and antisymmetric, representing the effective rotational speed that attempts to mimic the target Hamiltonian dynamics. The matrix $S$ is real, symmetric, and strictly negative semi-definite ($S \preceq 0$). This component drives the inevitable radial contraction, pulling the state away from the boundary of the Bloch sphere and representing irreversible decoherence. 

Since both $[\boldsymbol{\omega}\times]$ and $S$ emerge from \cref{dlm}, we prove that they have the relation 
\begin{equation}\label{maine1}
|\boldsymbol{\omega}|^2\le mC^2(-\Tr S).
\end{equation}
This indicates that one cannot increasing the rotational speed without increasing the radial decoherence. The $|\boldsymbol{\omega}|^2$ versus $-\Tr S$ have been seen in \cref{dmis} with the $\delta$-strength Hamiltonian part and the $\delta^2$-strength dissipative part. We can also understand this from \cref{fig2} left, where we can imagine the Hamiltonian dynamics via dissipation as rotating a top via a whip. Since the internal Hamiltonian is zero, top cannot rotate by itself, and can only rely on the external force of whip, therefore, it is impossible to rotate the top without radial disturbance.

Under the uniform-in-time requirement \cref{uerr}, we want the dissipative trajectory on the Bloch sphere to be close to the Hamiltonian trajectory as shown in \cref{fig2} right. This leads to two further consequences. First, $S$ should be small such that the decoherence won't make the radial distance between the Bloch surface and the Lindbladian evolved state from a pure initial state to go beyond $\epsilon$. We prove $S$ must satisfy
\begin{equation}\label{maine2}
-\Tr(S)\le \frac{6\epsilon}{T}.
\end{equation}
Second, $|T\boldsymbol{\omega}|$ should be large such that the length of the dissipative trajectory $l_D$ can catch the length of the Hamiltonian trajectory $l_H=\nu t$, which relates to the so-called quantum speed limit~\cite{taddei2013quantum,pires2016generalized,deffner2017quantum}. We prove that $|\boldsymbol{\omega}|$ must satisfy
\begin{equation}\label{maine3}
|\boldsymbol{\omega}|\ge \frac{\nu t}{8T}.
\end{equation}
By combining \cref{maine1}, \cref{maine2}, and \cref{maine3}, we proof \cref{the2}. Detailed proofs can be found in App. \ref{sec:uniform_time_lb}.

\noindent\textit{\textbf{Implications.—}}We now show several direct implications of our results.

\textit{1. Purely dissipative dynamics is BQP-complete.—}A direct implication of our result is that purely dissipative dynamics is BQP-complete~\cite{bernstein1993quantum} beyond the fixed point. Since Hamiltonian dynamics is BQP-complete, \cref{the1} indicates the purely dissipative Lindbladian is BQP-hard, which combining with the fact that Lindbladian can be efficiently simulated by repeated short-time unitary transformations plus tracing out operations via Stinespring dilation~\cite{nielsen2010quantum,kliesch2011dissipative,cleve2016efficient}, leads to the BQP-complete result. An important remark here is that, previously, the BQP-complete results~\cite{verstraete2009quantum,chen2024local,rouze2025efficient} of dissipation are all about the fixed point where the universal quantum computation is embedded into the steady states of the dissipative Lindbladian. Our result is different from the existing ones on BQP completeness of Lindbladians in the literature~\cite{verstraete2009quantum,chen2024local,rouze2025efficient} in the sense that we do not encode the computation into properties of a (quasi)stationary state of the dynamics. Other results~\cite{kastoryano2013precisely} achieve universality of Lindbladian dynamics by introducing auxiliary 'timer gadgets' to orchestrate sequential measurement-based operations. However, this typically requires significant spatial overhead in the form of complex lattice geometries and auxiliary clock qubits, overheads which are not present in our simple construction.

\textit{2. A novel Zeno-adjacent freezing mechanism.} Quantum Zeno effect~\cite{itano1990quantum} states that strong and frequent measurements (or jump operators with large norms in the language of Lindbladian) can sometimes freeze the system's evolution. The mechanism of the traditional Zeno effect is to construct strong dissipators $\mathcal{L}_D[\rho]=\sum_{i} \left(F_i\rho F_i^\dag- \frac{1}{2}\{\rho,F_i^\dag F_i\}\right)$, which is so strong, such that the system can ignore the internal Hamiltonian and is forced to stay in the steady subspace where $\mathcal{L}_D[\rho]=0$, also known as the Zeno subspace~\cite{facchi2002quantum}. Our \cref{the1} also leads to a fundamentally different mechanism of quantum Zeno effect. Suppose the system Hamiltonian is $H$, and we know $H$, then to freeze the dynamics, we can simply introduce a jump operator $F=\delta^{-1/2}(I+i\delta H)$. Then we have
\begin{equation}
\frac{d\rho}{dt}=-i[H,\rho]+ F\rho F^\dag- \frac{1}{2}\{\rho,F^\dag F\}=\mathcal{O}(\delta).
\end{equation}
Thus, setting $\delta=\epsilon/t$ can make sure $\|\rho(0)-\rho(t)\|_1=\mathcal{O}(\epsilon)$. The idea of this new mechanism is to actively generate the reversal time evolution via dissipation to cancel out the system evolution. While both require large dissipation, the biggest difference between these two mechanisms is that the traditional mechanism can only freeze states that are inside the Zeno subspace, while ours can freeze any state.

\textit{3. No super-quadratic fast forwarding for a class of Lindbladians.—}A central result in Hamiltonian simulation is the no fast-forwarding theorem~\cite{berry2007efficient,atia2017fast}, which states that there exist no universal (meaning for any $H$) Hamiltonian simulation algorithms that can simulate $t$-time Hamiltonian dynamics within a constant precision $\epsilon=1/4$ using sublinear costs in $t$. For the class of purely dissipative Lindbladian we considered ($F=I-i\delta H$), we can simply let $H$ be the sparse parity check Hamiltonian~\cite{berry2007efficient,farhi1998limit} used in the Hamiltonian no-fast-forwarding proof and set $\delta=\epsilon/t=1/4T$, then we have $T=\mathcal{O}(t^2)$. If this Lindbladian can be super-quadratically fast-forwarded, i.e. $o(\sqrt{T})$, we will conflict with the query lower bound of the parity check problem, implying this class of Lindbladians cannot be fast-forwarded beyond quadratic. On the other hand, this also indicates that this class of Lindbladians has the possibility to be quadratically fast-forwarded in some cases, which, if true, will be another class of fast-forwardable Lindbladians beyond the recent progress~\cite{shang2025fast,gao2025ancilla,shang2025exponential}. We want to emphasize that we cannot simply use the $t$-time Hamiltonian dynamics to simulate the $T$-time Lindbladian evolution and claim achieving the quadratic fast-forwarding since the $\delta^2$-strength of the dissipation is fixed and prevents the simulation error to be arbitrarly small. This raises the question of how to construct fast-forwarding algorithms for these Lindbladians.

\textit{4. Reducing Lindbladian simulation cost via gauge changing.—}One way to understand \cref{the1} is via the gauge freedom of Lindbladians. For $H$ and $F_i$ in \cref{mainlme}, we can do a gauge transform~\cite{breuer2002theory} while keeping the Liouvillian operator and therefore, the dynamics unchanged
\begin{align}
\text{Jump: }F_i&\rightarrow F_i'=\sum_{j} U_{ij}F_j+c_i I,\\
\text{Hamiltonian: }H&\rightarrow H'= -\frac{i}{2}\sum_{i}\left(c_i^*F_i'-c_iF_i'^\dag\right),
\end{align}
where $U$ is unitary, $c_i$ are arbitrary complex numbers. If we let $H=0$ and $F_i=I-i\delta H$, then by setting $U=I$ and $c_i=-1$ we will have $H'=\delta H$ and $F'=-i\delta H$, which recovers \cref{dmis}. In the context of simulating Lindbladian on quantum computers, given the Hamiltonian $H$ and jump operators $F_i$, the algorithmic cost is then proportional to $\|H\|_\infty+\sum_i\|F_i\|^2_\infty$~\cite{kliesch2011dissipative,childs2016efficient,cleve2016efficient,li2022simulating,yu2025lindbladian,peng2408quantum,ding2024simulating,pocrnic2025quantum,kato2024exponentially}. For the Lindbladian we considered in \cref{the1}, we have $\|F\|_\infty=\mathcal{O}(1)$ while $\|H'\|_\infty+\sum_i\|F'\|^2_\infty=\mathcal{O}(\delta)$. Therefore, while the gauge transform of the Lindbladian doesn't change the actual dynamics, it indeed can drastically change the simulation cost. This suggests that, if we want to simulate a Lindbladian on a quantum computer, instead of directly putting the Hamiltonian and jumps into the simulation algorithms, we can first find a better gauge with smaller norms and simulating the Lindbladian in that gauge will result in a lower cost.

\noindent\textit{\textbf{Discussions.—}}In this work, we show the Hamiltonian dynamics can be efficiently simulated via pure dissipation via explicit construction (\cref{the1}) and prove it matches the $\Omega(t^2/\epsilon)$ lower bound (\cref{the2}). 
It is worth contrasting our result with a line of work on emergent unitary dynamics from dissipation initiated by~\cite{zanardi2014coherent,zanardi2015geometry} and extended by~\cite{albert2016geometry}. These works consider a strongly dissipative Lindbladian $\mathcal{L}_0$ admitting a non-trivial manifold of steady states, perturbed by a weak Hamiltonian term $\lambda H$. In the limit $\lambda \to 0$ with time rescaled as $\tau = \lambda t$, the effective dynamics projected onto the steady-state manifold becomes unitary, generated by a projected Hamiltonian $PHP$. The setup we consider here is fundamentally different in three respects. First, our Lindbladian has no Hamiltonian term at all: the coherent dynamics is generated entirely by the dissipator. Second, our construction does not rely on the existence of a non-trivial steady-state manifold: the target unitary evolution is approximated on the entire Hilbert space, not on a projected subspace. Third, we provide a quantitative error bound with matching upper and lower bounds of order $t^2/\epsilon$, whereas the dissipation-projected dynamics literature is primarily concerned with the qualitative emergence of unitarity and its geometric structure. Our $\Omega(t^2/\epsilon)$ lower bound is also naturally dual to the $\mathcal{O}(t^2/\epsilon)$ overhead identified by~\cite{cleve2016efficient} for the reverse direction of simulating Lindblad evolution via Hamiltonian evolution on a dilated system, with a similar scaling appearing in the repeated-interactions simulation of~\cite{pocrnic2025quantum}.

Several directions can be further explored. For example, whether and when it is possible and how to achieve the quadratic fast-forwarding for our considered Lindbladians as explained in Implication 3. Also, while in Implication 4, we explain the importance of choosing gauge in Lindbladian simulations, accurately estimating $\|H\|_\infty+\sum_i\|F_i\|^2_\infty$ is actually QMA-hard~\cite{kempe2006complexity}. This raises the question of how to find a better gauge in a practical way.

\begin{acknowledgments}
Z.S. would like to thank Weijie Xiong and Wenjun Yu for insightful discussions.
D.S.F. acknowledges financial support from the Novo Nordisk Foundation (Grant No. NNF20OC0059939 Quantum for Life). D.S.F. and Z.S. acknowledge support by the ERC grant GIFNEQ 101163938.
\end{acknowledgments}

\bibliography{ref}

\begin{appendix}
\renewcommand\thefigure{\thesection.\arabic{figure}}
\onecolumngrid
\section{Proof of \cref{the1} and related results\label{proof}}
\textit{On $F=I-\delta A$:} To avoid ambiguity, in the following, we will use $\|\cdot\|_\infty$ to represent the spectral norm $\|\cdot\|$ in the main text. When we set the jump operator to be $F=I-\delta A$ with $A=G+iH$, we have
\begin{align}
F\rho F^\dag&=(I-\delta (G+iH))\rho(I-\delta (G-iH))\nonumber\\
&=\rho-\delta (G+iH)\rho-\delta \rho(G-iH)+\delta^2 (G+iH)\rho (G-iH)\nonumber\\
&=\rho-\delta\{G,\rho\}-i\delta [H,\rho]+\delta^2 (G+iH)\rho (G-iH).
\end{align}
Also, because we have
\begin{align}
F^\dag F&=(I-\delta (G-iH))(I-\delta (G+iH))\nonumber\\
&=I-\delta (G-iH)-\delta (G+iH)+\delta^2(G-iH)(G+iH)\nonumber\\
&=I-2\delta G+\delta^2(G-iH)(G+iH),
\end{align}
therefore, we have
\begin{align}
\frac{1}{2}\{F^\dag F,\rho\}&=\frac{1}{2}\{I-2\delta G+\delta^2(G-iH)(G+iH),\rho\}\nonumber\\
&=\rho-\delta\{G,\rho\}+\frac{1}{2}\delta^2\{(G-iH)(G+iH),\rho\}.
\end{align}
Thus, the derivative of the density matrix has the form
\begin{align}
\frac{d\rho}{dt}&=\mathcal{D}_F[\rho]=F^\dag \rho F-\frac{1}{2}\{F^\dag F,\rho\}\nonumber\\
&=\rho-\delta\{G,\rho\}-i\delta [H,\rho]+\delta^2 (G+iH)\rho (G-iH)-\rho+\delta\{G,\rho\}-\frac{1}{2}\delta^2\{(G-iH)(G+iH),\rho\}\nonumber\\
&=-i\delta [H,\rho]+\delta^2\mathcal{D}_A[\rho],
\end{align}
where we denote $\mathcal{D}_O[\rho]=O\rho O^\dag-\frac{1}{2}\{O^\dag O,\rho\}$.

Now, to bound the error of simulating the Hamiltonian dynamics, we first introduce the following lemma
\begin{lemma}\label{daml}
Given two Lindbladians with the Liouvillian operators $\mathcal{L}_1$ and $\mathcal{L}_2$ respectively, the diamond distance between the two Lindbladians for an evolution time $t$ has the upper bound
\begin{align}
\|e^{\mathcal{L}_1t}-e^{\mathcal{L}_2t}\|_\diamond\leq t\|\mathcal{L}_1-\mathcal{L}_2\|_\diamond.
\end{align}
\end{lemma}
\begin{proof}
Using the Duhamel Identity~\cite{evans2022partial}, we have
\begin{align}
e^{\mathcal{L}_1t} - e^{\mathcal{L}_2t} = \int_0^t e^{\mathcal{L}_2(t-s)} (\mathcal{L}_1-\mathcal{L}_2) e^{\mathcal{L}_1s} ds,
\end{align}
which leads to
\begin{align}
\|e^{\mathcal{L}_1t} - e^{\mathcal{L}_2t} \|_\diamond&\leq \int_0^t \|e^{\mathcal{L}_2(t-s)}\|_\diamond\| \mathcal{L}_1-\mathcal{L}_2\|_\diamond \|e^{\mathcal{L}_1s}\|_\diamond ds\nonumber\\
&\leq \int_0^t \| \mathcal{L}_1-\mathcal{L}_2\|_\diamond  ds=t\| \mathcal{L}_1-\mathcal{L}_2\|_\diamond,
\end{align}
where we use $\|e^{\mathcal{L}t}\|_\diamond\leq 1$.
\end{proof}
Now, using \cref{daml}, we have
\begin{align}
\|e^{\mathcal{D}_F t/\delta}-e^{\mathcal{L}_Ht}\|_\diamond\leq t\delta \|\mathcal{D}_A\|_\diamond,
\end{align}
where we use $\mathcal{L}_H$ to denote the Liouvillian operator of Hamiltonian dynamics under $H$. To have an approximation with an error $\epsilon$, we can simply ask
\begin{align}
t\delta \|\mathcal{D}_A\|_\diamond\leq \epsilon\rightarrow\delta\leq \frac{\epsilon}{t\|\mathcal{D}_A\|_\diamond}.
\end{align}

Since the Lindbladian evolution time $T$ is $T=t/\delta$, thus, we obtain
\begin{align}
T\geq  \|\mathcal{D}_A\|_\diamond t^2/\epsilon.
\end{align}
Since $\|\mathcal{D}_A\|_\diamond=\mathcal{O}(\|A\|_\infty^2)$, we finish the proof
\begin{align}
T=\mathcal{O}(\|A\|_\infty^2t^2/\epsilon).
\end{align}
When $A=iH$, we recover \cref{the1}.

\textit{On $F_i=I-i\delta H_i$:} Similarly, when $H=\sum_i H_i$, and set jump operators $F_i=I-i\delta H_i$, we have
\begin{align}
\frac{d\rho}{dt}&=\sum_i\mathcal{D}_{F_i}[\rho]=\sum_i\left(F_i^\dag \rho F_i-\frac{1}{2}\{F_i^\dag F_i,\rho\}\right)\nonumber\\
&=-i\delta \left[\sum_i H_i,\rho\right]+\sum_i\delta^2\mathcal{D}_{H_i}[\rho].
\end{align}
Using \cref{daml}, we have
\begin{align}
\|e^{\sum_i\mathcal{D}_{F_i} t/\delta}-e^{\mathcal{L}_Ht}\|_\diamond\leq t\delta \|\sum_i\mathcal{D}_{H_i}\|_\diamond\leq t\delta\sum_i\|\mathcal{D}_{H_i}\|_\diamond,
\end{align}
which leads to
\begin{align}
T=\mathcal{O}\left(\sum_i\|H_i\|_\infty^2t^2/\epsilon\right).
\end{align}

\textit{On General Lindbladian via dissipation:} To simulate the general Lindbladian
\begin{equation}
\frac{d\rho}{dt}=\mathcal{L}[\rho]=-i[H,\rho]+ \sum_{i} \left(F_i\rho F_i^\dag- \frac{1}{2}\{\rho,F_i^\dag F_i\}\right),
\end{equation}
via dissipation, we construct a purely dissipative Lindbladian $\mathcal{L}'$ with jumps $F'_0=I-i\delta H$ and $F_i'=\sqrt{\delta}F_i$. Since we have
\begin{align}
\frac{d\rho}{dt}&=\mathcal{L}'[\rho]=F_0'^\dag \rho F_0'-\frac{1}{2}\{F_0'^\dag F_0',\rho\}+\sum_i\left(F_i'^\dag \rho F_i'-\frac{1}{2}\{F_i'^\dag F_i',\rho\}\right)\nonumber\\
&=-i\delta [H,\rho]+\delta\sum_i\left(F_i^\dag \rho F_i-\frac{1}{2}\{F_i^\dag F_i,\rho\}\right)+\delta^2\mathcal{D}_{H}[\rho]\nonumber\\
&=\delta\mathcal{L}[\rho]+\delta^2\mathcal{D}_{H}[\rho],
\end{align}
therefore, following the similar procedures as above, we obtain
\begin{align}
T=\mathcal{O}\left(\|H\|_\infty^2t^2/\epsilon\right).
\end{align}

\section{Proof of \cref{the2}}
\label{sec:uniform_time_lb}

In this note we prove a rigorous lower bound under the \emph{stronger} simulation requirement that the dissipative evolution must approximate the target Hamiltonian evolution \emph{for every intermediate time} between $0$ and $t$, not just at the final time. The result is proved for a simple qubit family, which is enough to establish the desired $\Omega(t^2/\epsilon)$ scaling in the worst-case sense.

\paragraph{Simulation notion.}
Let
\begin{equation}
\mathcal{L}_D(\rho)=\sum_{j=1}^m\left(F_j\rho F_j^\dagger-\tfrac12\{F_j^\dagger F_j,\rho\}\right),
\qquad \|F_j\|_\infty\le C,
\end{equation}
be a purely dissipative Lindbladian on a qubit, and let $\Phi_\tau=e^{\tau\mathcal{L}_D}$. For a target Hamiltonian $H$, define the unitary channel
\begin{equation}
\mathcal{U}_s(\rho)=e^{-iHs}\rho e^{iHs}.
\end{equation}
We say that $\Phi_\tau$ simulates $\mathcal{U}_s$ up to time $t$ in total dissipative time $T$ if
\begin{equation}
\sup_{0\le s\le t}\left\|\Phi_{\frac{T}{t}s}-\mathcal{U}_s\right\|_\diamond\le \epsilon.
\label{eq:uniformsim}
\end{equation}
This is the natural uniform-in-time strengthening of the endpoint condition and matches the linear time reparametrization used by the constructive upper bound.

\begin{theorem}[Uniform-in-time lower bound for a hard qubit family]
\label{thm:uniform_qubit_lb}
Let
\begin{equation}
H=\frac{\nu}{2}\sigma_z,
\qquad \Theta:=|\nu| t.
\end{equation}
Assume that \eqref{eq:uniformsim} holds and that
\begin{equation}
0<\epsilon\le \frac{1}{100}\min\{1,\Theta\}.
\label{eq:eps_assumption}
\end{equation}
Then
\begin{equation}
T\ge \frac{\Theta^2}{384\,mC^2\,\epsilon}
=\frac{\nu^2 t^2}{384\,mC^2\,\epsilon}.
\label{eq:main_uniform_lb}
\end{equation}
Equivalently, since $\|H\|_\infty=|\nu|/2$,
\begin{equation}
T\ge \frac{\|H\|_\infty^2 t^2}{96\,mC^2\,\epsilon}.
\end{equation}
In particular, for any fixed-norm hard family with $\|H\|_\infty=\Theta(1)$, uniform-in-time pure-dissipative simulation requires
\begin{equation}
T=\Omega(t^2/\epsilon).
\end{equation}
\end{theorem}

\begin{proof}
We write the qubit state as
\begin{equation}
\rho=\frac12(I+\mathbf{r}\cdot\boldsymbol{\sigma}),
\qquad \mathbf{r}\in\mathbb{R}^3,
\end{equation}
and represent channels in Bloch form.

\paragraph{Step 1: Bloch form of a purely dissipative qubit generator.}
Write each jump operator as
\begin{equation}
F_j=a_j I+\mathbf{u}_j\cdot\boldsymbol{\sigma},
\qquad a_j\in\mathbb{C},\quad \mathbf{u}_j\in\mathbb{C}^3.
\end{equation}
A standard Pauli-basis calculation gives the affine Bloch equation~\cite{nielsen2010quantum}
\begin{equation}
\dot{\mathbf{r}}=B\mathbf{r}+\mathbf{c},
\qquad B=S+[\boldsymbol{\omega}\times],
\label{eq:bloch_generator}
\end{equation}
where $S=S^\top$ is real symmetric, $[\boldsymbol{\omega}\times]$ is the real antisymmetric matrix for the cross product by $\boldsymbol{\omega}$, and
\begin{align}
S&=2\sum_{j=1}^m\Bigl(\re(\mathbf{u}_j\mathbf{u}_j^\dagger)-|\mathbf{u}_j|^2 I_3\Bigr),
\label{eq:S_formula}\\
\boldsymbol{\omega}&=-2\sum_{j=1}^m \im(a_j^*\mathbf{u}_j).
\label{eq:omega_formula}
\end{align}
The translation vector $\mathbf{c}$ will not be needed explicitly.

Two consequences are immediate.
First, $S\preceq 0$. Indeed, for every real $\mathbf{x}\in\mathbb{R}^3$,
\begin{align}
\mathbf{x}^\top S\mathbf{x}
&=2\sum_{j=1}^m\Bigl(|\mathbf{x}\cdot\mathbf{u}_j|^2-|\mathbf{u}_j|^2|\mathbf{x}|^2\Bigr)\le 0.
\end{align}
Hence
\begin{equation}
\|S\|_\infty\le -\Tr(S).
\label{eq:Snorm_trace}
\end{equation}
Second, taking the trace in \eqref{eq:S_formula} gives
\begin{equation}
-\Tr(S)=4\sum_{j=1}^m |\mathbf{u}_j|^2.
\label{eq:traceS_formula}
\end{equation}
Using $|a_j|\le \|F_j\|_\infty\le C$ and Cauchy--Schwarz,
\begin{align}
|\boldsymbol{\omega}|
&\le 2\sum_{j=1}^m |a_j|\,|\mathbf{u}_j|
\le 2C\sqrt{m\sum_{j=1}^m |\mathbf{u}_j|^2} \\
&= C\sqrt{m(-\Tr S)},
\end{align}
where the last equality uses \eqref{eq:traceS_formula}. Thus
\begin{equation}
|\boldsymbol{\omega}|^2\le mC^2(-\Tr S).
\label{eq:omega_vs_traceS}
\end{equation}

\paragraph{Step 2: Uniform channel closeness implies uniform Bloch-matrix closeness.}
Let $A_\tau=e^{\tau B}$ denote the linear part of $\Phi_\tau$ in Bloch form, so that
\begin{equation}
\mathbf{r}\mapsto A_\tau\mathbf{r}+\mathbf{b}_\tau.
\end{equation}
The target channel generated by $H=\frac{\nu}{2}\sigma_z$ acts on Bloch vectors by the rotation
\begin{equation}
R_z(\alpha\tau),
\qquad \alpha:=\frac{\nu t}{T}=\frac{\Theta}{T}.
\end{equation}
Since the trace distance between qubit states equals the Euclidean distance between their Bloch vectors, \eqref{eq:uniformsim} implies that for every unit vector $\mathbf{r}\in\mathbb{R}^3$ and every $\tau\in[0,T]$,
\begin{equation}
\bigl|A_\tau\mathbf{r}+\mathbf{b}_\tau-R_z(\alpha\tau)\mathbf{r}\bigr|\le \epsilon.
\label{eq:statewise_bloch_close}
\end{equation}
Applying \eqref{eq:statewise_bloch_close} to both $\mathbf{r}$ and $-\mathbf{r}$ and adding/subtracting the two inequalities yields
\begin{equation}
\|A_\tau-R_z(\alpha\tau)\|_\infty\le \epsilon,
\qquad
|\mathbf{b}_\tau|\le \epsilon,
\qquad 0\le \tau\le T.
\label{eq:linear_and_translation_close}
\end{equation}

\paragraph{Step 3: Uniform closeness forces very small total contraction.}
Taking $\tau=T$ in \eqref{eq:linear_and_translation_close}, Weyl's inequality~\cite{roger1994topics} for singular values implies that every singular value of $A_T$ lies in $[1-\epsilon,1+\epsilon]$, because all singular values of the rotation $R_z(\Theta)$ are equal to $1$. Therefore
\begin{equation}
\det(A_T)\ge (1-\epsilon)^3.
\end{equation}
Since $A_T=e^{TB}$ and $\Tr([\boldsymbol{\omega}\times])=0$,
\begin{equation}
\det(A_T)=e^{T\Tr(B)}=e^{T\Tr(S)}.
\end{equation}
Hence, for $\epsilon\le 1/2$,
\begin{equation}
-\Tr(S)
\le \frac{3}{T}\log\frac{1}{1-\epsilon}
\le \frac{6\epsilon}{T}.
\label{eq:traceS_upper}
\end{equation}
By \eqref{eq:Snorm_trace},
\begin{equation}
\|S\|_\infty\le \frac{6\epsilon}{T}.
\label{eq:Snorm_upper}
\end{equation}

\paragraph{Step 4: Uniform closeness to the rotating orbit forces a macroscopic path length.}
Let
\begin{equation}
\mathbf{u}(\tau):=R_z(\alpha\tau)\mathbf{e}_x,
\qquad
\mathbf{v}(\tau):=A_\tau\mathbf{e}_x,
\end{equation}
where $\mathbf{e}_x=(1,0,0)^\top$. By \eqref{eq:linear_and_translation_close},
\begin{equation}
|\mathbf{v}(\tau)-\mathbf{u}(\tau)|\le \epsilon,
\qquad 0\le \tau\le T.
\label{eq:v_close_u}
\end{equation}
We claim that the Euclidean arc length $L(\mathbf{v})$ of $\mathbf{v}(\tau)$ satisfies
\begin{equation}
L(\mathbf{v})\ge \frac{\Theta}{5}.
\label{eq:pathlength_lower}
\end{equation}
Indeed:

\emph{Case 1: $0<\Theta\le 1$.}
Since length dominates endpoint distance,
\begin{align}
L(\mathbf{v})
&\ge |\mathbf{v}(T)-\mathbf{v}(0)|
\ge |\mathbf{u}(T)-\mathbf{u}(0)|-2\epsilon \\
&=2\sin(\Theta/2)-2\epsilon
\ge \Theta/2-2\epsilon
\ge \Theta/4,
\end{align}
where in the last step we used \eqref{eq:eps_assumption}.

\emph{Case 2: $\Theta\ge 1$.}
Let $N=\lfloor\Theta\rfloor$, so $N\ge \Theta/2$. For $k=0,1,\dots,N$, define $\tau_k=k/\alpha$. Then consecutive target points are separated by one radian, hence
\begin{equation}
|\mathbf{u}(\tau_{k+1})-\mathbf{u}(\tau_k)|=2\sin(1/2)>0.95.
\end{equation}
Using \eqref{eq:v_close_u} and \eqref{eq:eps_assumption},
\begin{equation}
|\mathbf{v}(\tau_{k+1})-\mathbf{v}(\tau_k)|
\ge |\mathbf{u}(\tau_{k+1})-\mathbf{u}(\tau_k)|-2\epsilon
>0.93.
\end{equation}
Therefore
\begin{equation}
L(\mathbf{v})\ge \sum_{k=0}^{N-1}|\mathbf{v}(\tau_{k+1})-\mathbf{v}(\tau_k)|
>0.93 N
\ge 0.46\Theta
>\frac{\Theta}{5}.
\end{equation}
This proves \eqref{eq:pathlength_lower}.

\paragraph{Step 5: Convert path length into a lower bound on the effective rotation rate.}
Since $\dot{\mathbf{v}}(\tau)=B\mathbf{v}(\tau)$, we have
\begin{align}
L(\mathbf{v})
&=\int_0^T |\dot{\mathbf{v}}(\tau)|\,d\tau
\le \int_0^T \|B\|_\infty\,|\mathbf{v}(\tau)|\,d\tau \\
&\le (1+\epsilon)T\bigl(|\boldsymbol{\omega}|+\|S\|_\infty\bigr),
\label{eq:pathlength_upper}
\end{align}
where we used $\|[\boldsymbol{\omega}\times]\|_\infty=|\boldsymbol{\omega}|$ and $|\mathbf{v}(\tau)|\le |\mathbf{u}(\tau)|+\epsilon=1+\epsilon$.
Combining \eqref{eq:pathlength_lower}, \eqref{eq:pathlength_upper}, and \eqref{eq:Snorm_upper}, we obtain
\begin{equation}
|\boldsymbol{\omega}|
\ge \frac{\Theta}{5(1+\epsilon)T}-\frac{6\epsilon}{T}.
\end{equation}
Under \eqref{eq:eps_assumption}, the second term is at most $0.06\,\Theta/T$ and the first term is at least $0.19\,\Theta/T$, so in particular
\begin{equation}
|\boldsymbol{\omega}|\ge \frac{\Theta}{8T}.
\label{eq:omega_lower}
\end{equation}

\paragraph{Step 6: Finish the lower bound.}
Combining \eqref{eq:omega_lower}, \eqref{eq:omega_vs_traceS}, and \eqref{eq:traceS_upper}, we find
\begin{equation}
\frac{\Theta^2}{64T^2}
\le |\boldsymbol{\omega}|^2
\le mC^2(-\Tr S)
\le \frac{6mC^2\epsilon}{T}.
\end{equation}
Rearranging gives
\begin{equation}
T\ge \frac{\Theta^2}{384\,mC^2\,\epsilon},
\end{equation}
which is exactly \eqref{eq:main_uniform_lb}.
\end{proof}

\section{Uniqueness of construction for smooth Lindbladians}\label{app:uniqueness}
In this section we discuss to what extent our construction is unique, showing that any smooth family of semigroups that can approximate Hamiltonian dynamics to arbitrary precision (potentially at the cost of increased simulation time), must have a form similar to our evolution up to second order.
\begin{prop}[Uniqueness up to $O(\delta^2)$]\label{prop:rigidity}
Let $\{\mathcal{L}_{D,\delta}\}_{0\le \delta<\delta_0}$ be a smooth family of purely dissipative Lindbladians on a finite-dimensional system, and let
\begin{equation}
\mathcal{L}_H(\rho):=-i[H,\rho].
\end{equation}
Assume that there exists $t_0>0$ such that
\begin{equation}
\varepsilon(\delta):=
\sup_{0\le s\le t_0}
\left\|
e^{\frac{s}{\delta}\mathcal{L}_{D,\delta}}
-
e^{s\mathcal{L}_H}
\right\|_\diamond
\xrightarrow[\delta\to 0]{}0.
\label{eq:rigidity_uniform_assumption}
\end{equation}
Then
\begin{equation}
\mathcal{L}_{D,\delta}
=
\delta \mathcal{L}_H + O(\delta^2)
\label{eq:rigidity_generator}
\end{equation}
in diamond norm. In particular, the first-order term is uniquely fixed by $H$. Equivalently, if $\widetilde{\mathcal{L}}_{D,\delta}$ is another smooth purely dissipative family satisfying \eqref{eq:rigidity_uniform_assumption} for the same Hamiltonian $H$, then
\begin{equation}
\mathcal{L}_{D,\delta}-\widetilde{\mathcal{L}}_{D,\delta}=O(\delta^2).
\end{equation}

Moreover, assume that for some fixed $m$ one can choose jump operators smoothly so that
\begin{equation}
\mathcal{L}_{D,\delta}(\rho)
=
\sum_{j=1}^m
\left(
F_{j,\delta}\rho F_{j,\delta}^\dag
-\frac12\{F_{j,\delta}^\dag F_{j,\delta},\rho\}
\right),
\qquad
F_{j,\delta}=F_{j,0}+\delta G_j+O(\delta^2).
\end{equation}
Then $F_{j,0}=c_j I$ for some scalars $c_j$, and
\begin{equation}
H
=
\sum_{j=1}^m
\frac{c_j G_j^\dag-c_j^*G_j}{2i}.
\label{eq:rigidity_first_order_jump_constraint}
\end{equation}
Hence, modulo the usual Lindblad gauge freedom and up to $O(\delta^2)$, every nontrivial first-order jump has the form
\begin{equation}
F_{j,\delta}
=
\sqrt{\gamma_j}\bigl(I-\delta A_j\bigr)+O(\delta^2),
\qquad
\sum_j \gamma_j\frac{A_j-A_j^\dag}{2i}=H,
\label{eq:rigidity_jump_normal_form}
\end{equation}
while jumps with $F_{j,0}=0$ contribute only at order $O(\delta^2)$.
\end{prop}

\begin{proof}
Set
\begin{equation}
\Phi_{\delta,s}:=e^{\frac{s}{\delta}\mathcal{L}_{D,\delta}},
\qquad
U_s:=e^{s\mathcal{L}_H}.
\end{equation}

We first show that $\mathcal{L}_{D,0}=0$. Fix any $a>0$. For all sufficiently small $\delta$, we have $a\delta\le t_0$, and therefore \eqref{eq:rigidity_uniform_assumption} implies
\begin{equation}
\left\|
e^{a\mathcal{L}_{D,\delta}}
-
e^{a\delta\mathcal{L}_H}
\right\|_\diamond
\le
\varepsilon(\delta).
\end{equation}
Now let $\delta\to 0$. Since $\mathcal{L}_{D,\delta}$ is smooth at $\delta=0$, we get
\begin{equation}
e^{a\mathcal{L}_{D,0}}=I.
\end{equation}
Since this holds for every $a>0$, differentiating at $a=0$ yields
\begin{equation}
\mathcal{L}_{D,0}=0.
\end{equation}

Because the family is smooth, there exists a bounded superoperator $\mathcal{M}$ such that
\begin{equation}
\mathcal{L}_{D,\delta}
=
\delta \mathcal{M}+O(\delta^2),
\qquad
\mathcal{M}
=
\partial_\delta \mathcal{L}_{D,\delta}\big|_{\delta=0}.
\end{equation}
Define
\begin{equation}
A_\delta:=\delta^{-1}\mathcal{L}_{D,\delta}
=
\mathcal{M}+O(\delta).
\end{equation}
By Duhamel's formula,
\begin{equation}
e^{sA_\delta}-e^{s\mathcal{M}}
=
\int_0^s
e^{(s-u)A_\delta}(A_\delta-\mathcal{M})e^{u\mathcal{M}}
\,du.
\end{equation}
Hence, for $0\le s\le t_0$,
\begin{align}
\|e^{sA_\delta}-e^{s\mathcal{M}}\|_\diamond
&\le
s\,e^{s(\|A_\delta\|_\diamond+\|\mathcal{M}\|_\diamond)}
\|A_\delta-\mathcal{M}\|_\diamond \nonumber\\
&=O(\delta),
\end{align}
uniformly on $[0,t_0]$. Since $e^{sA_\delta}=e^{\frac{s}{\delta}\mathcal{L}_{D,\delta}}$, combining this with \eqref{eq:rigidity_uniform_assumption} gives
\begin{align}
\sup_{0\le s\le t_0}\|e^{s\mathcal{M}}-e^{s\mathcal{L}_H}\|_\diamond\leq \|e^{s\mathcal{M}}-e^{\frac{s}{\delta}\mathcal{L}_{D,\delta}}\|_\diamond+\|e^{\frac{s}{\delta}\mathcal{L}_{D,\delta}}-e^{s\mathcal{L}_H}\|_\diamond=O(\delta+\epsilon(\delta)).
\end{align}
Letting $\delta\to0$ we conclude that:
\begin{equation}
\sup_{0\le s\le t_0}
\|e^{s\mathcal{M}}-e^{s\mathcal{L}_H}\|_\diamond
=0.
\end{equation}
Therefore
\begin{equation}
e^{s\mathcal{M}}=e^{s\mathcal{L}_H}
\qquad\text{for all }0\le s\le t_0.
\end{equation}
Differentiating at $s=0$ yields
\begin{equation}
\mathcal{M}=\mathcal{L}_H.
\end{equation}
This proves \eqref{eq:rigidity_generator}.

We now prove the jump statement. Write
\begin{equation}
\mathcal{D}_F(\rho):=F\rho F^\dag-\frac12\{F^\dag F,\rho\},
\qquad
\mathcal{L}_{D,\delta}=\sum_{j=1}^m \mathcal{D}_{F_{j,\delta}}.
\end{equation}
Since $\mathcal{L}_{D,0}=0$, we have
\begin{equation}
0=\mathcal{L}_{D,0}=\sum_{j=1}^m \mathcal{D}_{F_{j,0}}.
\end{equation}
Let $\rho=|\psi\rangle\langle\psi|$ be pure. Then
\begin{align}
0
&=
\Tr(\rho\,\mathcal{L}_{D,0}[\rho]) \nonumber\\
&=
-\sum_{j=1}^m
\left(
\langle\psi|F_{j,0}^\dag F_{j,0}|\psi\rangle
-
|\langle\psi|F_{j,0}|\psi\rangle|^2
\right).
\end{align}
Each summand is nonnegative, so every summand must vanish for every $|\psi\rangle$. Hence each $F_{j,0}$ has zero variance on every pure state, which implies
\begin{equation}
F_{j,0}=c_j I
\end{equation}
for some scalar $c_j$.

Now write
\begin{equation}
F_{j,\delta}=c_j I+\delta G_j+O(\delta^2).
\end{equation}
A direct expansion gives
\begin{align}
\mathcal{D}_{F_{j,\delta}}[\rho]
&=
\delta
\left(
c_j^*G_j\rho+c_j\rho G_j^\dag
-\frac12\{c_jG_j^\dag+c_j^*G_j,\rho\}
\right)
+O(\delta^2)
\nonumber\\
&=
-i\delta
\left[
\frac{c_jG_j^\dag-c_j^*G_j}{2i},
\rho
\right]
+O(\delta^2).
\label{eq:one_jump_rigidity_expansion}
\end{align}
Summing \eqref{eq:one_jump_rigidity_expansion} over $j$ and comparing with \eqref{eq:rigidity_generator} proves
\begin{equation}
H
=
\sum_{j=1}^m
\frac{c_jG_j^\dag-c_j^*G_j}{2i},
\end{equation}
which is \eqref{eq:rigidity_first_order_jump_constraint}.

Finally, if $c_j\neq 0$, write
\begin{equation}
G_j=-c_jA_j,
\end{equation}
so that
\begin{equation}
F_{j,\delta}=c_j(I-\delta A_j)+O(\delta^2).
\end{equation}
The overall phase of $c_j$ is irrelevant because $\mathcal{D}_{e^{i\theta}F}=\mathcal{D}_F$. Thus, after removing the phase and defining $\gamma_j:=|c_j|^2$, we may rewrite
\begin{equation}
F_{j,\delta}=\sqrt{\gamma_j}(I-\delta A_j)+O(\delta^2).
\end{equation}
Then \eqref{eq:rigidity_first_order_jump_constraint} becomes
\begin{equation}
\sum_j \gamma_j\frac{A_j-A_j^\dag}{2i}=H.
\end{equation}
If $c_j=0$, then $F_{j,\delta}=O(\delta)$ and hence
\begin{equation}
\mathcal{D}_{F_{j,\delta}}=O(\delta^2),
\end{equation}
so such jumps do not contribute at first order. This proves \eqref{eq:rigidity_jump_normal_form}.
\end{proof}

\end{appendix}
\end{document}